\documentclass[final,authoryear,10pt]{elsarticle2}

\usepackage{amssymb}
\usepackage[hyphens]{url}
\usepackage[breaklinks=true]{hyperref}
\usepackage{lineno}

\usepackage{ctable}
\usepackage{tabularx}
\newcolumntype{+}{>{\global \let \currentrowstyle \relax}}
\newcolumntype{^}{>{\currentrowstyle }}

\usepackage[cmex10]{amsmath}

\hypersetup{linktocpage=true, urlcolor=magenta,linkcolor=blue,citecolor=blue, colorlinks=true,bookmarks=true,bookmarksopen=true,bookmarksopenlevel=0}
\newtheorem{thm}{Theorem}[section]

\newenvironment{proof}{%
{\noindent \bf Proof. }%
}{%
\hfill$\Box$\\%
}

\def\R{\mathbb{R}}

\usepackage{lscape}

\begin{document}

\begin{frontmatter}

\title{Accounting for Symptomatic and Asymptomatic in a SEIR-type model of COVID-19}

\author[csu]{Jayrold P. Arcede}
\author[ili]{Randy L. Caga-anan}
\author[lamfa,csu]{Cheryl Q. Mentuda}
\author[lamfa]{Youcef Mammeri\corref{cor1}}
\address[lamfa]{Laboratoire Ami\'enois de Math\'ematique Fondamentale et Appliqu\'ee, CNRS UMR 7352, Universit\'e de Picardie Jules Verne, 80069 Amiens, France}
\address[csu]{Department of Mathematics, Caraga State University, Butuan City, Philippines}
\address[ili]{Department of Mathematics and Statistics, MSU-Iligan Institute of Technology, Iligan City, Philippines}

\cortext[cor1]{Corresponding author: youcef.mammeri@u-picardie.fr}

\begin{abstract}
A mathematical model was developed describing the dynamic of the COVID-19 virus over a population considering that the infected can either be symptomatic or not. The model was calibrated using data on the confirmed cases and death from several countries like France, Philippines, Italy, Spain, United Kingdom, China, and the USA. First, we derived the basic reproduction number, $\mathcal{R}_0$, and estimated the effective reproduction $\mathcal{R}_{\rm eff}$ for each country. 
Second, we were interested in the merits of interventions, either by distancing or by treatment. Results revealed that total and partial containment is effective in reducing the transmission. However, its duration may be long to eradicate the disease (104 days for France). By setting the end of intervention as the day when hospital capacity is reached, numerical simulations showed that the duration can be reduced (up to only 39 days for France if the capacity is 1000 patients).  Further, results pointed out that the effective reproduction number remains large after containment. Therefore, testing and isolation are necessary to stop the disease.
\end{abstract}

\begin{keyword}

    COVID-19 \sep SE$I_s$$I_a$UR Model \sep Reproduction number \sep Interventions 

\MSC[2010] 92D30\sep 37N25 \sep 34D20
\end{keyword}

\end{frontmatter}


\section{Introduction}
In late 2019, a disease outbreak emerged in a city of Wuhan, China.  The culprit was a certain strain called  Coronavirus Disease 2019 or COVID-19 in brief \citep{WHO2020a}.  This virus has been identified to cause fever, cough, shortness of breath, muscle ache, confusion, headache, sore throat, rhinorrhoea, chest pain, diarrhea, and nausea and vomiting \citep{Hui2020, Chen2020}. COVID-19 belongs to the \textit{Coronaviridae} family. A family of coronaviruses that cause diseases in humans and animals, ranging from the common cold to more severe diseases.  Although only seven coronaviruses are known to cause disease in humans,  three of these,  COVID-19 included, can cause a much severe infection, and sometimes fatal to humans. The other two to complete the list were the severe acute respiratory syndrome (SARS) identified in 2002 in China, and the Middle East respiratory syndrome (MERS) originated decade after in Saudi Arabia. 

Like MERS and SARS, COVID-19 is a zoonotic virus and believed to be originated from bats transmitted to humans \citep{Zhou2020}. 
In comparison with SARS, MERS, the COVID-19 appears to be less deadly. However, the World Health Organization (WHO)  reported that it has already infected and killed more people than its predecessors combined.   Also, COVID-19 spreads much faster than SARS and MERS. It only took over a month before it surpassed the number of cases recorded by the SARS outbreak in 2012. According to WHO, it only took  67 days from the beginning of the outbreak in China last December 2019 for the virus to infect the first 100,000 people worldwide \citep{WHO2020b}.
As of the 25th of March 2020, a cumulative total of 372,757 confirmed cases, while 16,231 deaths have been recorded for COVID-19 by World Health Organization \citep{WHO2020c}.

Last 30th of January, WHO characterized COVID-19 as Public Health Emergency of International Concern (PHEI) and urge countries to put in place strong measures to detect disease early, isolate and treat cases, trace contacts, and promote social distancing measures commensurate with the risk \citep{WHO2020d}. In response, the world implemented its actions to reduce the spread of the virus. Limitations on mobility, social distancing, and self-quarantine have been applied.  Moreover,  health institutions advise people to practice good hygiene to keep from being infected.   All these efforts have been made to reduce the transmission rate of the virus.

For the time being,  COVID-19 infection is still on the rise. Government and research institutions scramble to seek antiviral treatment and vaccines to combat the disease.  Several reports list possible drugs combination to apply, yet it is still unclear which drugs could combat the viral disease and which won't.  

Several mathematical models have been proposed from various epidemiological groups. These models help governments as an early warning device about the size of the outbreak, how quickly it will spread, and how effective control measures may be. However, due to the limited emerging understanding of the new virus and its transmission mechanisms, the results are largely inconsistent across studies. 

In this paper, we will mention a few models and, in the end, to propose one. Gardner and his team \citep{Gardner2020} at  Center for Systems Science and Engineering, Johns Hopkins University, implemented on a previously published model applied for COVID-19. It is a metapopulation network model represented by a graph with each nodes follows a  discrete-time Susceptible-Exposed-Infected-Recovered (SEIR) compartmental model. The model gives an estimate of the expected number of cases in mainland China at the end of January 2020, as well as the global distribution of infected travelers. They believe that the outbreak began in November 2019  with hundreds of infected already present in Wuhan last early December 2019.  
\cite{Wu2020} from WHO Collaborating Centre for
Infectious Disease Epidemiology and Control at the University of Hong Kong presented a modeling study on the nowcast and forecast of the 2019-nCoV outbreak at Wuhan.  The group used an SEIR metapopulation model to simulate epidemic and found reproductive number  $R_0$ around 2.68 (with 95\%
credible interval 2.47-2.86).  \cite{Imai2020} estimates  $R_0$ around 2.6 with uncertainty range of 1.5-3.5.  \cite{Zhao2020a, Zhao2020b} found $R_0$ to range from 2.24 to 5.71 based on the reporting rate of cases. If the reporting rate   increase 2-fold, $R_0=3.58$, if it increase 8-fold, $R_0=2.24$. If there is no change in the reporting rate, the estimated $R_0$ is 5.71. 
Similar to the above authors,  \cite{Wang2020}  employed an infectious disease dynamics model (SEIR model) for modeling and predicting the number of COVID-19 cases in Wuhan. They opined that to reduce $R_0$ significantly, the government should continue implementing strict measures on containment and public health issues.  
In the same tune as the latter, the model of \cite{dan2020} also suggests continuing to implement effective quarantine measures to avoid a resurgence of infection. The model consists of five (5) compartments:  susceptible, infected, alternative infection, detected, and removed.

Here, we proposed an extension from the classical SEIR model by adding a compartment of asymptomatic infected.  We address the challenge of predicting the spread of COVID-19 by giving our estimates for the basic reproductive numbers $\mathcal{R}_0$ and its effective reproductive number $\mathcal{R}_{\rm eff}$.  Afterward, we also assess risks and interventions via containment strategy or treatments of exposed and symptomatic infected.

The rest of the paper is organized as follows. Section 2, outlines our methodology. Here the model was explained, where the data was taken, and its parameter estimates.  Section 3 contained the qualitative analysis for the model. Here, we give the closed-form equation of reproductive number $\mathcal{R}_0$, then tackling the best strategy to reduce transmission rates.  Finally, section 4 outlines our brief discussion on some measures to limit the outbreak.

\section{Materials and methods}

\subsection{Confirmed and death data}
In this study, we used the publicly available dataset of COVID-19 provided by the Johns Hopkins University \citep{Dong2020}. This dataset includes many countries' daily count of confirmed cases, recovered cases, and deaths. Data can be downloaded from 
{\small \url{https://github.com/CSSEGISandData/COVID-19/tree/master/csse_covid_19_data}.}  These data are collected through public health authorities' announcements and are directly reported public and unidentified patient data, so ethical approval is not required.

\subsection{Mathematical model}

\begin{figure}[htbp]
\centering
\includegraphics[scale=0.5]{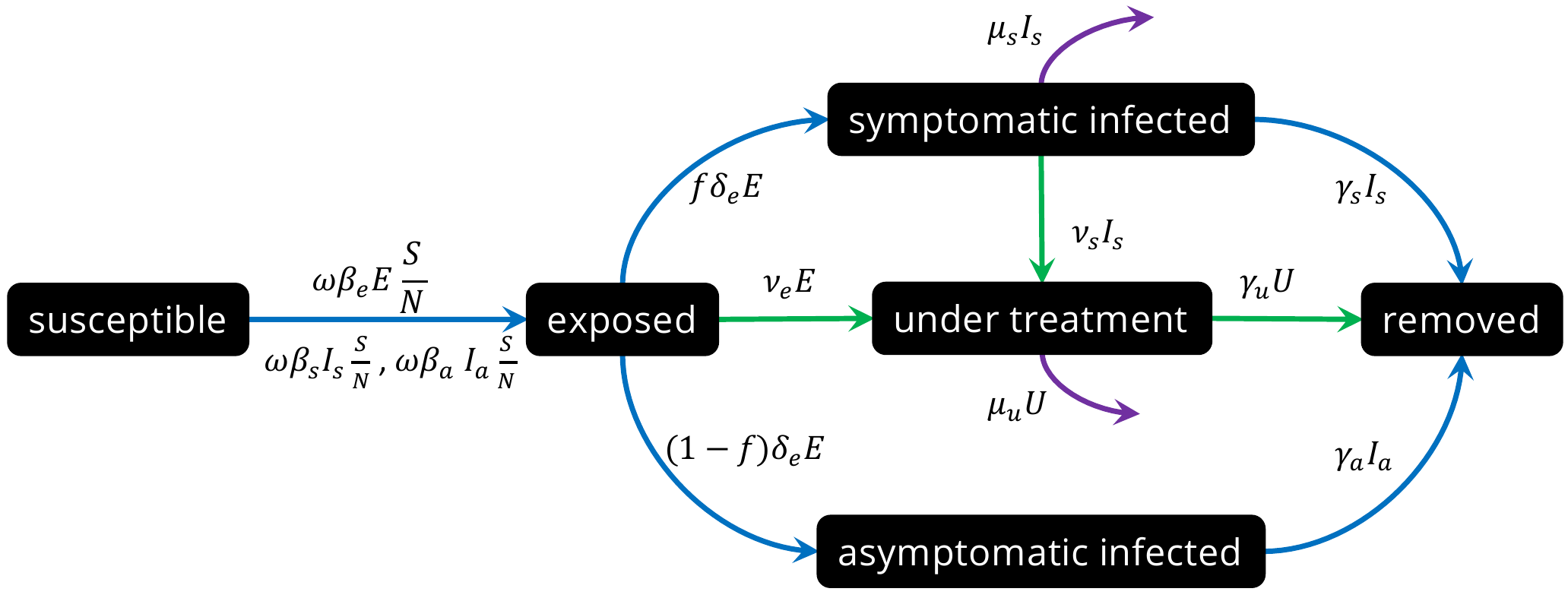} 
\\[-0.25cm]\caption{Compartmental representation of the $SEI_aI_sUR-$model. 
Blue arrows represent the infection flow. Green arrows denote for the treatments. Purple arrow is the death.}
\label{fig0}
\end{figure}

We focus our study on six components of the epidemic flow (Figure \ref{fig0}), {\it i.e.} Susceptible individual ($S$),
Exposed individual ($E$), Symptomatic Infected individual ($I_s$), Asymptomatic Infected individual ($I_a$), Under treatment individual ($U$), and Removed individual ($R$). 
To build the mathematical model, we followed the standard strategy developed in the literature concerning SIR model \citep{odi2000,bra2012}. We assumed that susceptible can be infected by exposed, symptomatic infected as well as asymptomatic infected individuals. We supposed that exposed and symptomatic infected can be treated without distinguishing quarantine or hospitalization. 
The dynamics is governed by a system of six ordinary differential equations (ODE) as follows
\begin{eqnarray*}
  S'(t) &=& - \omega \left( \beta_e E + \beta_s I_s + \beta_a I_a \right) \frac{S}{N} 
    \\
  E'(t) &=&    \omega\left( \beta_e E + \beta_s I_s + \beta_a I_a \right) \frac{S}{N} - (\delta_e + \nu_e) E 
    \\
  I_s'(t) &=& f\delta_e E - (\gamma_s + \mu_s + \nu_s ) I_s
    \\
  I_a'(t) &=& (1-f) \delta_e E - \gamma_a I_a
  \\
  U'(t) &=& \nu_e E + \nu_s I_s - (\gamma_u +\mu_u) U 
    \\
  R'(t) &=& \gamma_s I_s + \gamma_a I_a + \gamma_u U
\end{eqnarray*}
Note that the total living population follows $
    N'(t) =  - \mu_s I_s - \mu_u U $,
while death is computed by
$
    D'(t) = \mu_s I_s + \mu_u U.
$
We assume that there is no new recruit.
The parameters are described in Table A of Figure \ref{fig1}.

\subsection{Parameters estimation}

Calibration is made before intervention. Thus it is set $\nu_e = \nu_s = 0$.
The model is made up of eleven parameters $\theta = (\omega,\beta_e,\beta_s,\beta_a,\delta_e,f,\gamma_s,\mu_s,\gamma_a, \gamma_u,\mu_u)$ that need to be determined. Given, for $N$ days, the observations 
$I_{s, obs}(t_i)$ and $D_{obs}(t_i)$, the cost function consists of a nonlinear least square function
$$
J(\theta) = \sum_{i=1}^N (I_{s, obs}(t_i) - I_s(t_i,\theta))^2 +  (D_{obs}(t_i) - D(t_i,\theta))^2  ,
$$
with constraints $\theta \geq 0$, and $0 \leq f \leq 1$.
Here $I_s(t_i,\theta)$ and $D(t_i,\theta)$ denote output of the mathematical model at time $t_i$ computed with the parameters $\theta$.  The optimization problem is solved using Approximate Bayesian Computation combined with a quasi-Newton method \citep{Csi2010}.

\section{Results}

\subsection{Basic and effective reproduction numbers}

It is standard to check that the domain
$$ \Omega = \{ (S, E, I_s, I_a, U, R) \in \R_+^6; 0 \leq S + E + I_s + I_a + U + R \leq N(0) \}
$$
is positively invariant. In particular, 
there exists a unique global in time solution $(S, E, I_s, I_a, U, R)$ in $\mathcal{C} (\R_+; \Omega)$ as soon as the initial condition lives in $\Omega$.
\\
Since the infected individuals are in $E, I_a$ and $I_s$, 
the rate of  new infections in each compartment ($\mathcal{F}$)  and  the rate of other transitions between compartments 
($\mathcal{V}$) can be rewritten  as 
$$\mathcal{F} =\begin{pmatrix} 
 \omega ( \beta_e E + \beta_s I_s + \beta_a I_a ) \frac{S}{N} \\
0 \\ 
0 
 \end{pmatrix}, 
\; \mathcal{V} = \begin{pmatrix} 
(\delta_e +\nu_e) E \\
(\gamma_s + \mu_s + \nu_s) I_s - f \delta_e E\\
 \gamma_a I_a - (1- f) \delta_e E
\end{pmatrix}.
$$
Thus,
 $$
F = \begin{pmatrix}
 \frac{  \omega\beta_e S }{N} &  \frac{ \omega \beta_s S }{N} &  \frac{ \omega\beta_a S }{N} \\
  0 & 0 & 0 \\
  0 & 0 & 0
 \end{pmatrix},$$
and 
$$
V= \begin{pmatrix}  (\delta_e +\nu_e)   & 0 & 0 \\
 -f \delta_e & \gamma_s + \mu_s + \nu_s & 0 \\
 -(1 - f) \delta_e & 0 & \gamma_a
  \end{pmatrix}, \;
V^{-1}= \begin{pmatrix}  
\frac{1}{\delta_e +\nu_e} & 0 & 0 \\
\frac{f}{\gamma_s + \mu_s + \nu_s } & \frac{1}{\gamma_s + \mu_s + \nu_s} & 0 \\
\frac{1 -f }{ \gamma_a } & 0 & \frac{1}{ \gamma_a}
  \end{pmatrix}.
  $$
  Therefore,  the next generation matrix is
  \begin{align*}
 FV^{-1} &= \begin{pmatrix} 
 \frac{  \omega\beta_e S }{ (\delta_e +\nu_e)  N } + \frac{ f  \omega\beta_s S }{ (\gamma_s + \mu_s + \nu_s) N } + \frac{ ( 1 - f)  \omega\beta_a S }{ \gamma_a N} &  \frac{  \omega\beta_s S }{ (\gamma_s + \mu_s + \nu_s) N } &  \frac{ \omega\beta_a S }{ \gamma_a N} \\
  0 & 0 & 0 \\
  0 & 0 & 0
 \end{pmatrix}.
\end{align*}
We deduce as in \cite{vdd2000} that the basic reproduction number $\mathcal{R}_0$  for the
Disease Free Equilibrium $(S^*, 0, 0, 0, R^*)$, with $N^* = S^* + R^*$, is
$$
\mathcal{R}_0 :=  \omega\left( \frac{ \beta_e  }{ \delta_e +\nu_e  } + \frac{ f \beta_s  }{ \gamma_s + \mu_s + \nu_s } + \frac{ ( 1 - f) \beta_a  }{ \gamma_a } \right) \frac{S^*}{N^*} =: \mathcal{R} \frac{S^*}{N^*}. 
$$
This number has an epidemiological meaning. The term $ \frac{ \omega \beta_e  }{ \delta_e +\nu_e  }$ represents the contact rate with exposed during the average latency period  $1/(\delta_e +\nu_e)$. The term $\frac{ \omega f \beta_s  }{ \gamma_s + \mu_s + \nu_s }$ is the contact rate with symptomatic during the average infection period, and the last one is the part of asymptomatic.

In the subsequent, we write \textit{DFE} when we mean by Disease Free Equilibrium.
\begin{thm}
The DFE $(S^*, 0, 0, 0, R^*)$ is the unique positive equilibrium. Moreover it is globally asymptotically stable.
\end{thm} 

\begin{proof}
By computing the eigenvalues of the Jacobian matrix, we deduce that if $\mathcal{R}_0 < 1$,
then DFE is locally asymptotically stable.

Here we will prove that global asymptotic stability is independent that of $\mathcal{R}_0$. Indeed, from the last differential equation in our system of ODE, we can deduce that
$R$ is an increasing function bounded by $N(0)$. Thus $R(t)$ converges to $R^*$ as $t$ goes to $+\infty$.
Then integrating over time this equation provides
$$
    R(t) - R(0) =  \int_0^t \gamma_s I_s(s) + \gamma_a I_a(s) + \gamma_u U(s) \, ds
$$
and
$$
    R^* - R(0) = \gamma \int_0^{+\infty} \gamma_s I_s(s) + \gamma_a I_a(s) + \gamma_u U(s) \, ds, 
$$
which is finite. Furthermore, $I_s(t), I_a(t), U(t)$ also go to $0$ as  $t\to +\infty$
thanks to the positivity of the solution. Similarly, integrating the third or the fourth equation in the system gives   $E(t) \to_{t\to +\infty} 0$. Finally, the first equation points to the decreasing of $S$ that is bounded by below by $0$ and $S(t) \to_{t\to +\infty} S^*$.
\end{proof} 

This theorem means that the asymptotic behavior does not depend on $\mathcal{R}_0$. 
For all initial data in $\Omega$, the solution converges to the DFE when time goes to infinity.
Nevertheless, to observe initial exponential growth, $\mathcal{R}_0>1$ is necessary. Indeed, $S$ is initially close to $N$ such that infected states
are given by the linear system of differential equations
\[
  \begin{pmatrix}
    E \\ I_s \\ I_a
  \end{pmatrix}'(t) =   \begin{pmatrix} \omega \beta_e -\delta_e -\nu_s & \omega \beta_s & \omega \beta_a \\
            f\delta_e & -(\gamma_s+\mu_s+\nu_s) & 0
            \\ (1-f)\delta_e & 0 & -\gamma_a
     \end{pmatrix}   \begin{pmatrix}
    E \\ I_s \\ I_a
  \end{pmatrix}. 
\]
The characteristic polynomial  is $P(x) = x^3 + a_2 x^2 + a_1 x + a_0$,  with $a_0 =  (\gamma_s+\mu_s+\nu_s)\gamma_a (\delta_e +\nu_s)  \left( 1- \mathcal{R} \right)$.
If $\mathcal{R}>1$, there is at least one positive eigenvalue that coincides with an initial exponential growth rate of solutions.

To better reflect the time dynamic of the disease, the effective reproduction number  
$$ 
\mathcal{R}_{\rm eff}(t) =  \omega\left( \frac{ \beta_e  }{ \delta_e +\nu_s  } + \frac{ f \beta_s  }{ \gamma_s + \mu_s + \nu_s } + \frac{ ( 1 - f) \beta_a  }{ \gamma_a } \right) \frac{S(t)}{N(t)} .
$$
is represented in Figure \ref{fig1}D and values of $\mathcal{R}$ are computed in Table A of Figure \ref{fig1}.

\subsection{Model resolution}

To calibrate the model, our simulations start the day of first confirmed infection and finish before interventions to reduce the disease.
Therefore $\nu_e$ and $\nu_s $ are assumed to be equal to $0$. We assume that the whole population of the country is susceptible to the infection. Seven states with comparable populations are chosen.
The objective function $J$ is computed to provide a relative error of order less than $10^{-2}$.
In Figure \ref{fig1}, Table A. shows estimated parameters.
The rest of the Figure presents the solution and data. Note that the product $\omega \beta_i$ ($i=e,s,a$) is uniquely identifiable but not $\omega$ and $\beta_i$ separately. Figure \ref{fig2} represents the effective reproduction number, the fitted symptomatic infected and death of the posterior distribution.

\begin{figure}[htbp]
\centering
  \includegraphics[scale=.65]{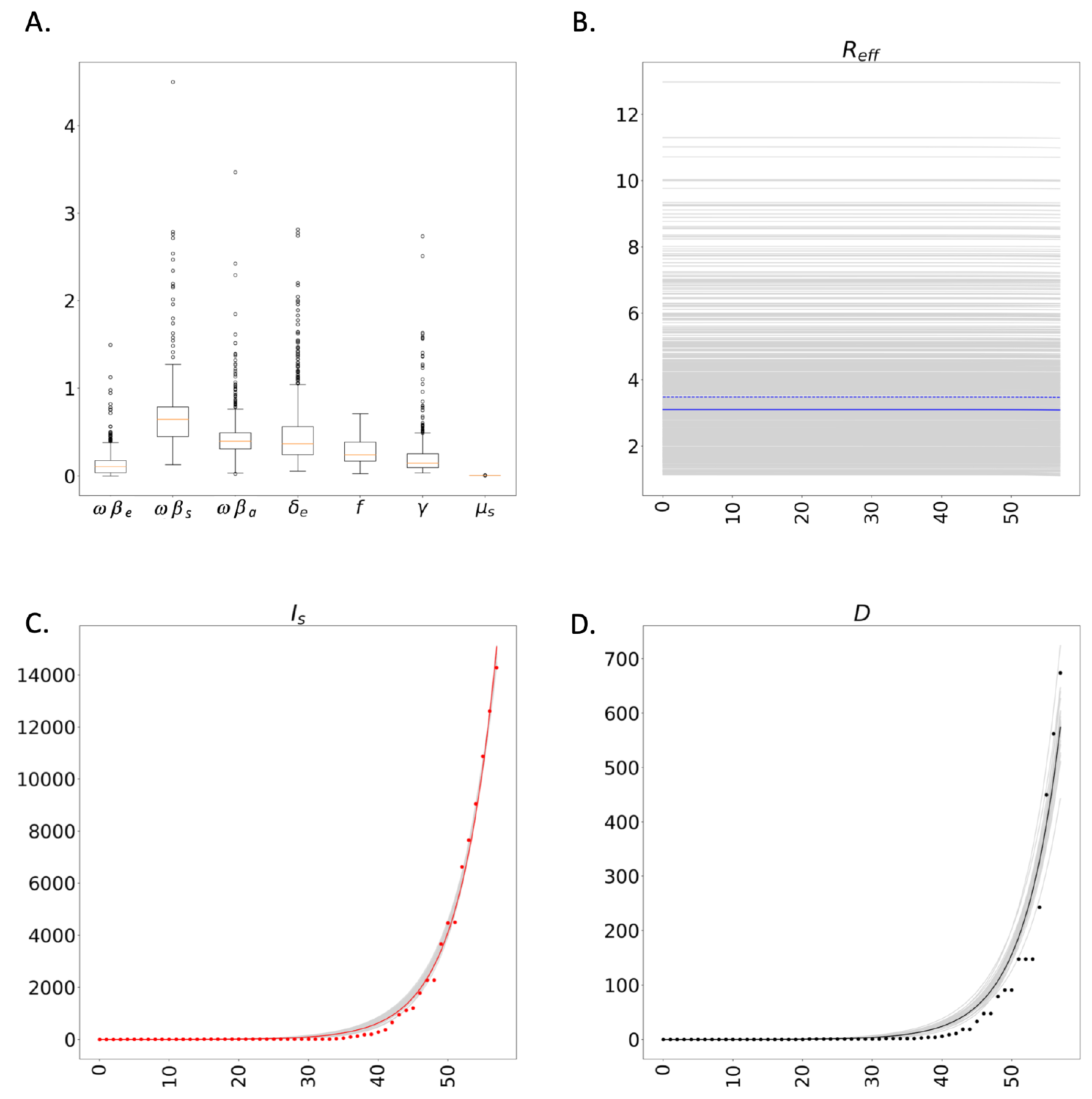}
\\[-0.5cm]\caption{A. Boxplot of the posterior distribution computed from France data. B. Effective reproduction number in grey of the posterior distribution, median ($= 3.096738 $) in blue straight line, mean ($= 3.474858 $) is dotted line. C. Fitted symptomatic infected in grey of the posterior distribution, median in red straight line, mean is dotted line. D. Fitted death in grey of the posterior distribution, median in black straight line, mean is dotted line.}
\label{fig2}
\end{figure}

\begin{landscape}

\begin{figure}[htbp]
\centering
  \includegraphics[scale=.7]{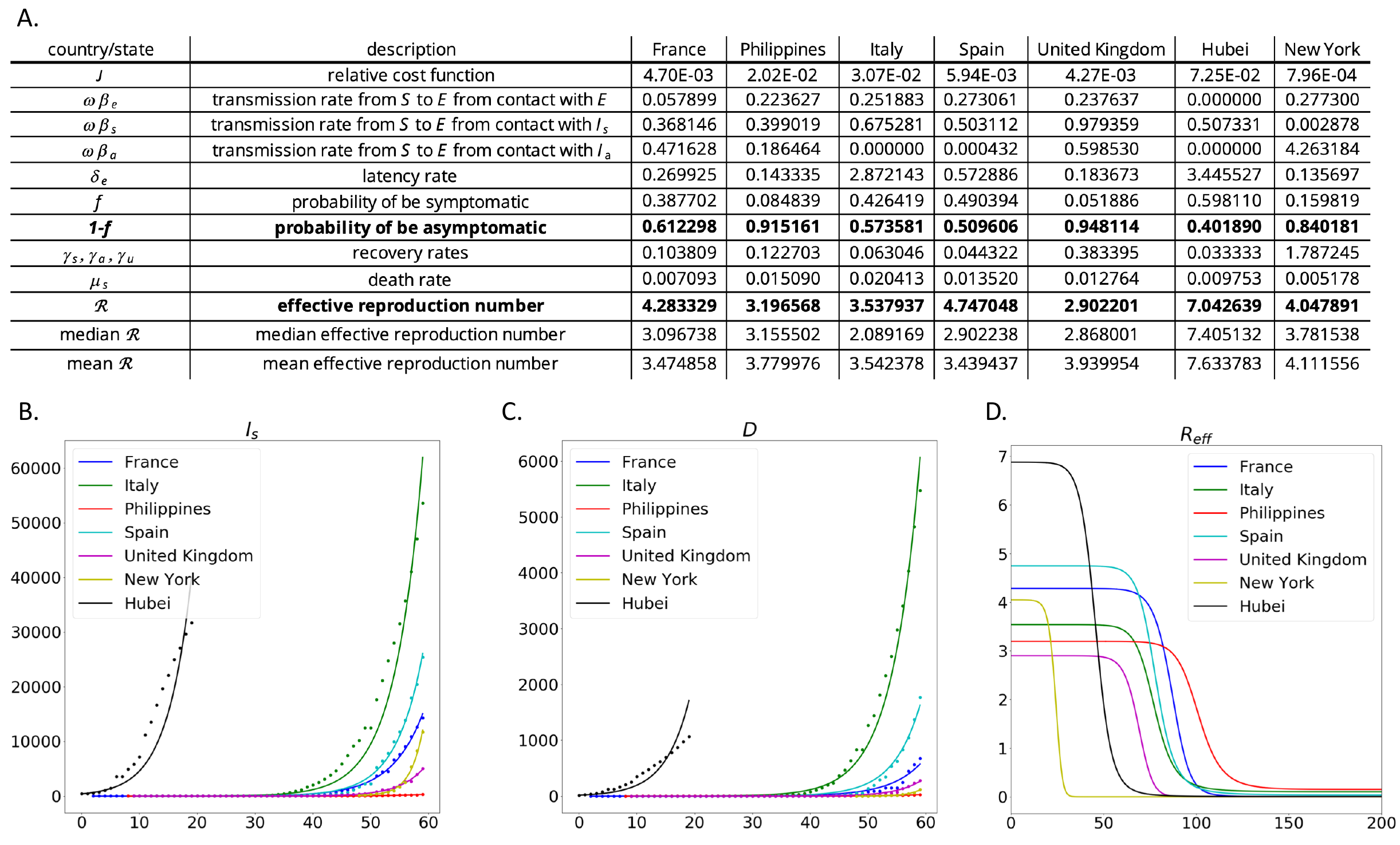}
\\[-0.5cm]\caption{A. Parameters calibrated according to data  from France, Philippines, Italy, Spain, United Kingdom, Hubei, and New York.  B. and C. Calibrated solution (straight line) and data (dots) with respect to day for France, Philippines, Italy, Spain, United Kingdom, Hubei and New York. First is the infected $I_s$ (B.), and the second one is the death $D$ (C.). D. Effective reproduction number with respect to day.}
\label{fig1}
\end{figure}

\end{landscape}

\newpage

\subsection{Strategy to reduce disease to a given threshold}

To temporarily reduce the value of $\mathcal{R}_{\rm eff}$, three strategies regarding different interventions are being compared. 
The first one consists in reducing the number of contacts $\omega$. Full containment is translated by $\omega=0$, and no containment by $\omega=1$.
The second strategy is expressed by treatment of symptomatic infected. It is translated by modifying the value of $\nu_s$.
The treatment of exposed is the third one. Parameters  $\omega, \nu_e$, and $\nu_s$ vary such that intervention is carried out from $0$ to $100\%$ of susceptible, exposed, and symptomatic respectively.
Given a critical infection threshold $\mathcal{T}_c$, the model runs until the time $t_c$  such that 
$$ E(t_c) + I_s(t_c) + I_a(t_c) \leq \mathcal{T}_c. $$

To juxtapose the benefit of the intervention, we assume that interventions start $53$ days after the first confirmed infection.
We remind that the first infection in France was confirmed on January $24^{\rm th}$, 2020, and containment begins on March $17^{\rm th}$, 2020.
Comparison between three strategies can be found in Table \ref{tab2} and Figures \ref{fig3}-\ref{fig3b}. Without intervention to control the disease, the maximum number of symptomatic infecteds varies from $3.49 \times 10^5$
to $2.02 \times 10^7$. The maximum number of deaths totals from $8.85 \times 10^3$ to $7.92 \times 10^6$.
We observe that any intervention strongly reduces the number of dead. Concerning France, Philippines, Italy, Spain, and the United Kingdom, when containment is fully respected and when the sum of infecteds is reduced to $1$,  the maximum number of symptomatic infecteds and deaths
has been cut sharply, of order $10^3$. It varies now from $5.75 \times 10^2$
to $7.04 \times 10^4$ and $1.94 \times 10^2$
to $2.52 \times 10^4$ respectively. To wait from $104$ to $407$ days is the 
price to pay. On the contrary, for the states of Hubei and New York,  $53$ days to intervene seems to be already too late.
We can also see  in Table \ref{tab2} that treating only the symptomatic does not reduce the duration
\\
Note that when the intervention ends at time $t_c$, the number of
susceptible $S(t_c)$ is large so that the effective reproduction number 
$\mathcal{R}_{\rm eff}$ is larger than $1$. 

Figure \ref{fig4} compares the maximum number of dead and symptomatic but infected individuals, as well as the intervention duration, to reach $\mathcal{T}_c=1000$ varying from $0$ to $100\%$.  Containment is the most efficient when it is respected by more than $76\%$ in France, $63\%$ in the Philippines. Beyond that, treating the exposed is the best choice. We also observe that the intervention duration becomes long below $89\%$ in France, $82\%$ in the Philippines. This can be understood by too little susceptibility to achieve recovery but enough for the disease to persist.

\begin{table}[htbp]
\begin{center}
\includegraphics[scale=.7]{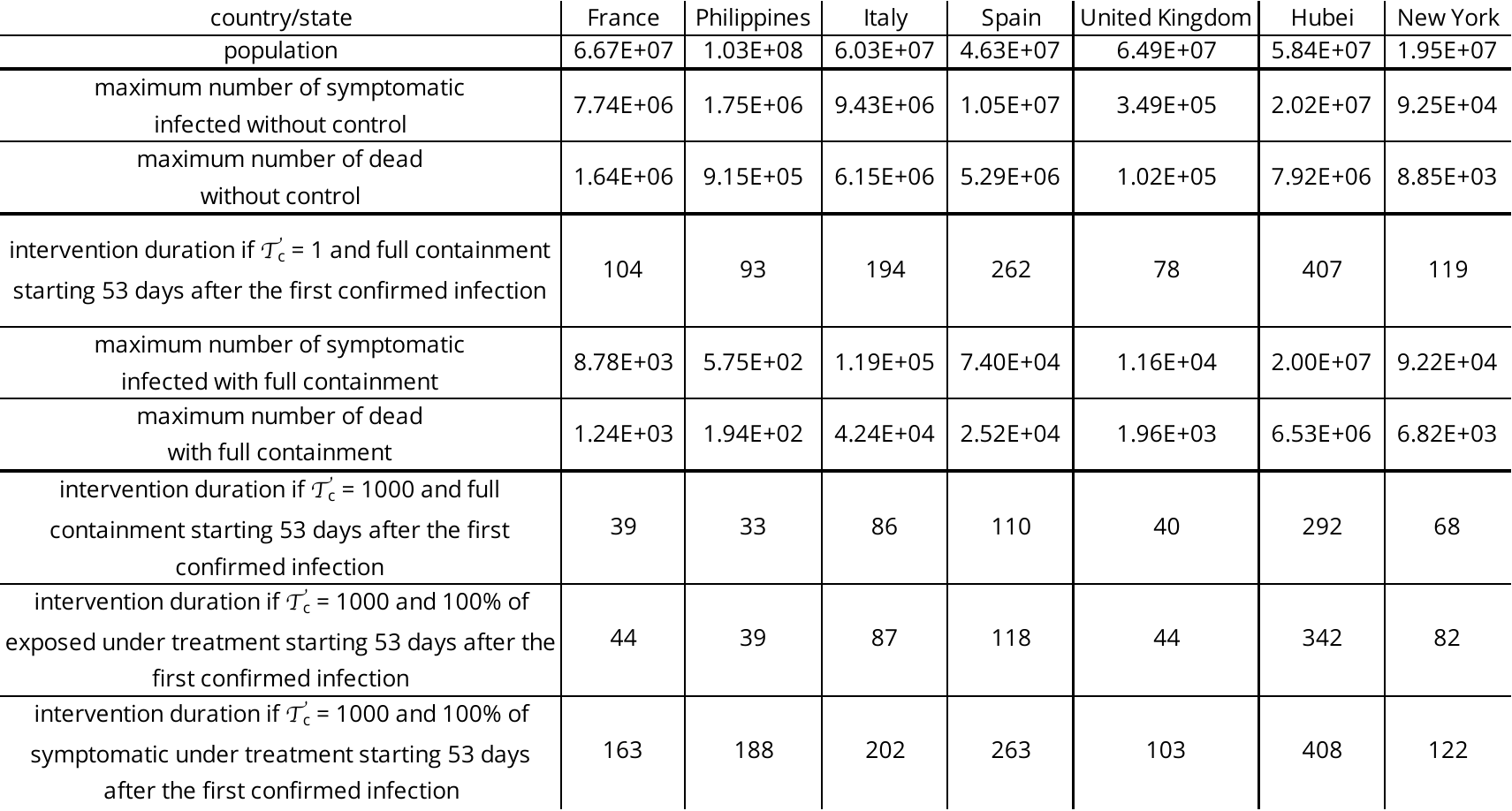}
\caption{Comparison between the maximum number of symptomatic infected and death without control and the solution reducing contact rate to $0$, $100\%$ of exposed under treatment, and $100\%$ of symptomatic infected under treatment to reach $\mathcal{T}_c = 1$ and $\mathcal{T}_c = 1000$. Interventions are assumed to being $53$ days after the first confirmed infection.}
\label{tab2}
\end{center}
\end{table}

\newpage

\begin{landscape}

\begin{figure}[htbp]
\centering
 \includegraphics[scale=.65]{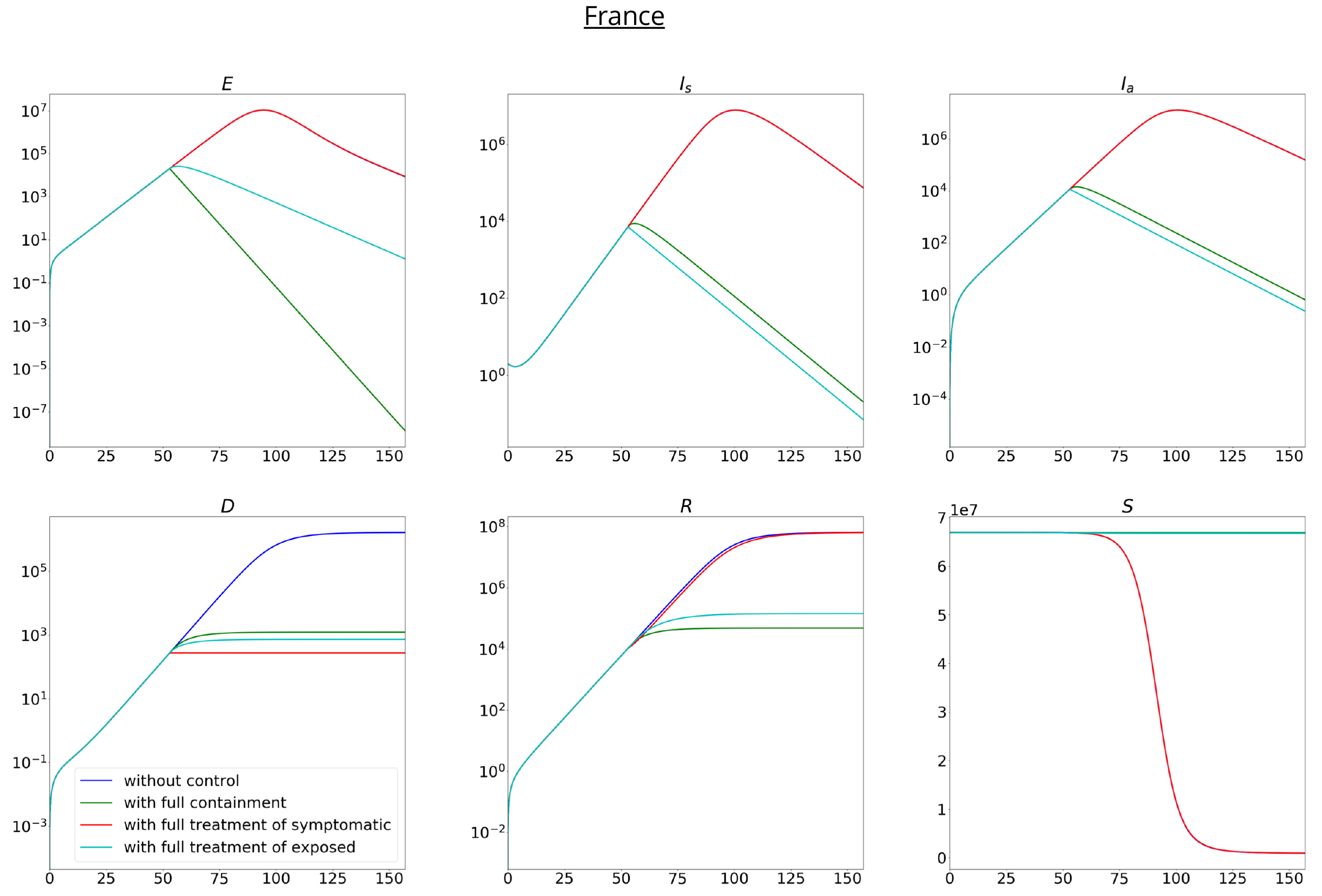}
\\[-0.5cm]\caption{Comparison of solutions $S, E, I_s, I_a, R, D$ without control in blue, full containment in green,
 full treatment of symptomatic in red, and full treatment of exposed in cyan for France. Ordinate axis is expressed in log.}
\label{fig3}
\end{figure}

\end{landscape}

\newpage

\begin{landscape}

\begin{figure}[htbp]
\centering
 \includegraphics[scale=.65]{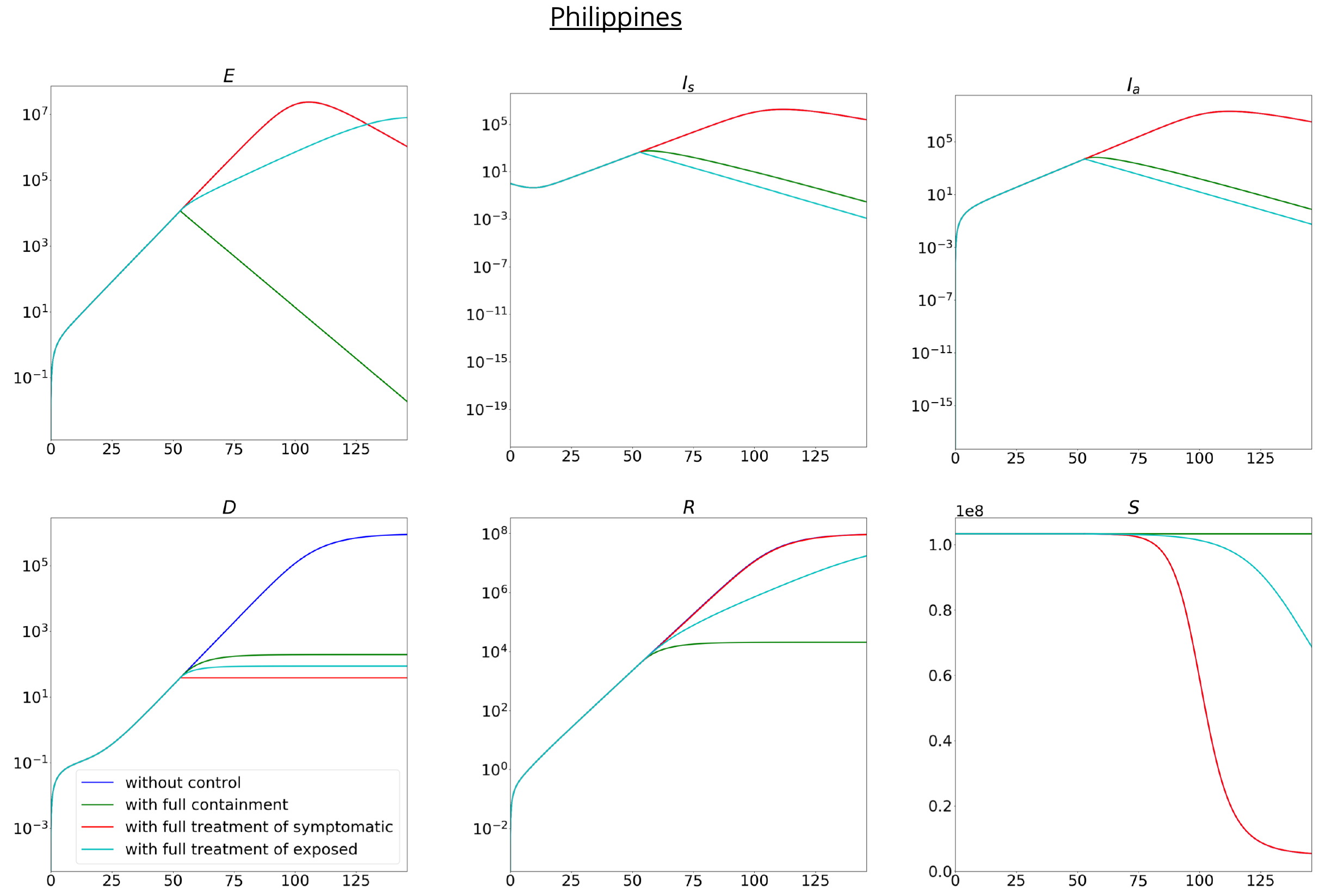}
\\[-0.5cm]\caption{Comparison of solutions $S, E, I_s, I_a, R, D$ without control in blue, full containment in green,
 full treatment of symptomatic in red, and full treatment of exposed in cyan for the Philippines. Ordinate axis is expressed in log.}
\label{fig3b}
\end{figure}

\end{landscape}

\newpage
\begin{landscape}
\begin{figure}[htbp]
\centering
 \includegraphics[scale=.7]{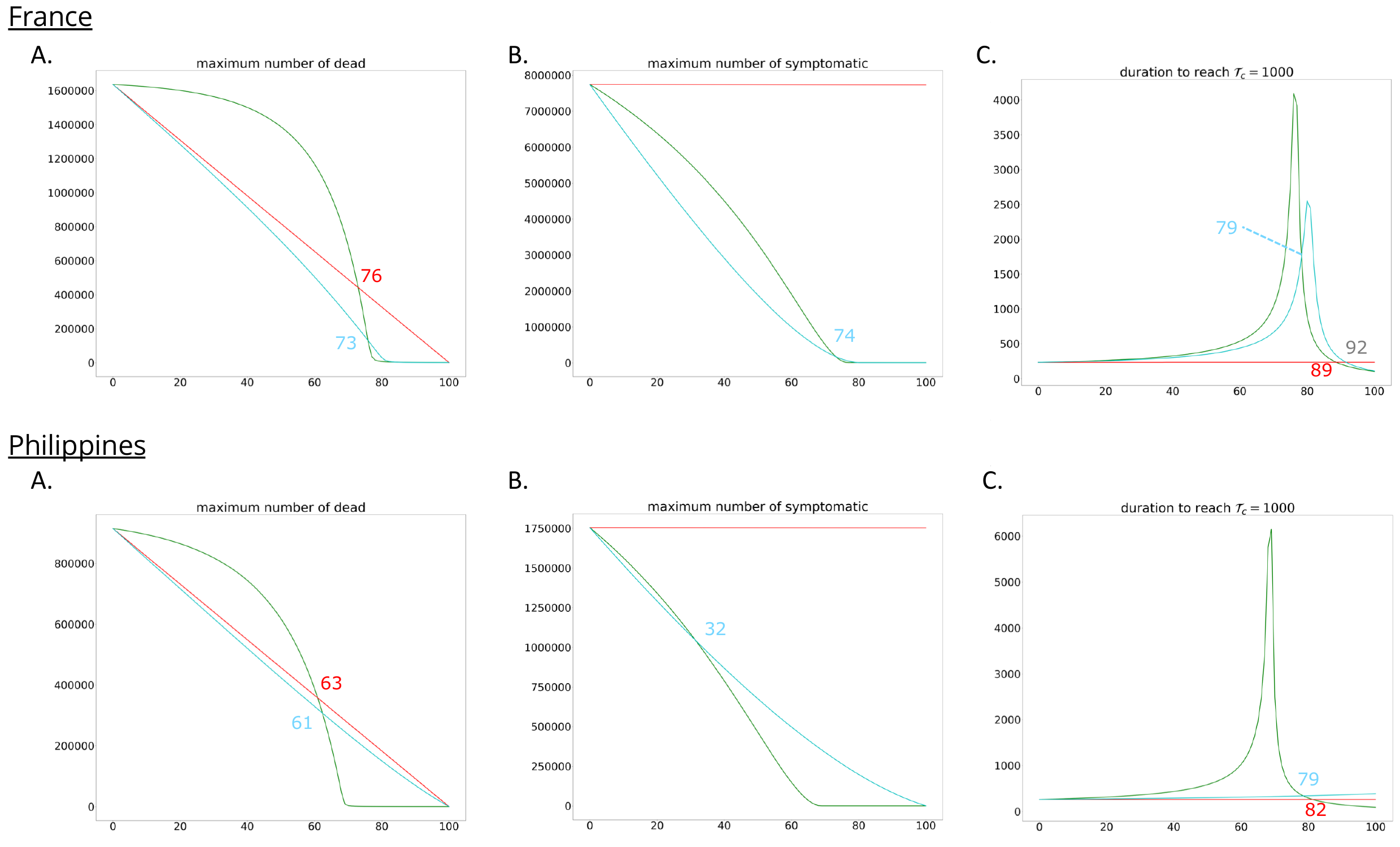}
 \\[-0.5cm]\caption{A. Comparison of the maximum number of dead.  B. Comparison of the maximum number of symptomatic infected. C. Comparison of the intervention duration to reach $\mathcal{T}_c = 1000$ with respect to percentage of containment (green),
  treatment of symptomatic (red), and  treatment of exposed (cyan) for France and Philippines.}
\label{fig4}
\end{figure}

\end{landscape}

\section{Discussion}
Without intervention, we observe in Figures \ref{fig3}-\ref{fig3b} that the number of susceptible $S$ is decreasing; most of the individuals are recovering, which generates population immunity. It translates that the disease free equilibrium is globally asymptotically stable. 
Nevertheless, the price to pay is high, the number of deaths being excessive.
As presented in Figure \ref{fig1}, the effective reproduction is decreasing and points out that control has to be done as fast as possible.

The other important information is that, as discovered by \cite{dan2020}, an alternative transmission way may occur.
Here, it is due to the proportion of asymptomatic infected individuals that is not negligible, as shown in Table \ref{tab2}.

Finally, with the little knowledge about COVID-19 nowadays, decreasing transmission, {\it i.e.} $\beta_e, \beta_s, \beta_a$, is the preferred option. The simplest choice consists in reducing contact between individuals. 
 Table \ref{tab2} and Figures \ref{fig3}-\ref{fig3b}-\ref{fig4} show that total and partial containment do indeed drastically reduce the disease.
However, the duration of containment may be too long and then impracticable especially if we aim at totally eradicating the infection ($\mathcal{T}_c = 1$).
Instead, to stop the containment as soon as the capacity of the hospitals has been reached could be privileged. When this criterion is set to 1000 patients ($\mathcal{T}_c = 1000$), the duration goes from $104$ to $39$ days for France. A similar reduction in duration is also obtained for other countries. Again, we see that the earlier the intervention, the more effective it is. Due to the high number of susceptible, it is worth noting that the effective reproduction number remains large after containment. Screening tests, especially to carry out exposed individuals, are then necessary to be carried out, and the positive individuals are quarantined.

\bibliographystyle{abbrvnat}
{
\bibliography{references}
}


\end{document}